\documentclass[a4paper]{article}

\usepackage[margin=35mm]{geometry}
\usepackage{microtype}
\usepackage{amsmath,amsthm,amsfonts,mathtools}
\usepackage{enumitem}
\usepackage{natbib}
\usepackage{tikz}
\usepackage{wrapfig}
\usepackage[hidelinks]{hyperref}

\usetikzlibrary{graphs,arrows}

\newlist{compactitem}{itemize}{2}
\setlist[compactitem]{label=\textbullet,noitemsep,leftmargin=*}
\setlist*[compactitem,2]{label=\textendash,nosep}

\theoremstyle{plain}
\newtheorem{theorem}{Theorem}

\newtheorem{corollary}[theorem]{Corollary}

\theoremstyle{definition}
\newtheorem*{definition}{Definition}

\DeclareMathOperator{\bigO}{\mathnormal{O}}
\DeclareMathOperator{\poly}{poly}

\newcommand{\algorithmicand}{\textbf{and}}
\newcommand{\algorithmicor}{\textbf{or}}
\newcommand{\bbN}{\mathbb{N}}
\newcommand{\binary}{{\{\bits0, \bits1\}}}
\newcommand{\bits}{\texttt}
\newcommand{\cl}[1]{\textnormal{\textbf{#1}}} % Complexity classes
\newcommand{\cladv}[2]{{\textnormal{\textbf{#1}}\!_\textnormal{\textbf{/#2}}}}
\newcommand{\length}[1]{{\lvert#1\rvert}}
\newcommand{\pr}[1]{\textnormal{\textsc{#1}}} % Problems
\newcommand{\size}[1]{{\lvert#1\rvert}}
\newcommand{\suchthat}[1][]{\;#1|\;} % Such that (with optional size)

\title{A Hierarchy of Polynomial Kernels\footnote{\href{https://doi.org/10.1007/978-3-030-10801-4_39}{doi:10.1007/978-3-030-10801-4\_39}}}
\author{
       Jouke Witteveen%\thanks{
  \and Ralph Bottesch\thanks{
This author was supported by the ERC Consolidator Grant QPROGRESS~615307 for the majority of the project's duration, and by the Austrian Science Fund~(FWF) project~Y757 at the time of publication.
  }
  \and Leen Torenvliet%\footnotemark[1]
}

\begin{document}

\maketitle

\begin{abstract}
  In parameterized algorithmics, the process of \emph{kernelization} is defined as a polynomial time algorithm that transforms the instance of a given problem to an equivalent instance of a size that is limited by a function of the parameter.
  As, afterwards, this smaller instance can then be solved to find an answer to the original question, kernelization is often presented as a form of \emph{preprocessing}.
  A natural generalization of kernelization is the process that allows for a number of smaller instances to be produced to provide an answer to the original problem, possibly also using negation.
  This generalization is called \emph{Turing} kernelization.
  Immediately, questions of equivalence occur or, when is one form possible and not the other.
  These have been long standing open problems in parameterized complexity.
  In the present paper, we answer many of these.
  In particular, we show that Turing kernelizations differ not only from regular kernelization, but also from intermediate forms as truth-table kernelizations.
  We achieve absolute results by diagonalizations and also results on natural problems depending on widely accepted complexity theoretic assumptions.
  In particular, we improve on known lower bounds for the kernel size of compositional problems using these assumptions.
\end{abstract}

\section{Introduction}

\paragraph*{Fixed-parameter tractability.}
For many important computational problems, the best known algorithms have a worst-case running time that scales exponentially or worse with the size of the input.
Generally however, the size of an input instance is a poor indicator of whether the instance is indeed difficult to solve.
This is because for most natural problems, a good fraction of all instances of a given size can be solved much more efficiently than the worst-case instance of that size.
To gain a better understanding of the \emph{complexity of individual instances}, we might define a function $\kappa:\binary^\ast\rightarrow \bbN$ that assigns to each instance $x$ a numeric \emph{parameter} $\kappa(x)$.
This parameter then indicates the extent to which certain features that we have identified as a potential cause of computational hardness are present in the given instance.
If the function $\kappa$ is itself polynomial-time computable, we call it a \emph{parameterization}.
We shall assume that $\kappa(x) \leq \length{x}$ holds for all $x \in \binary^\ast$.

Consider a problem for which the fastest known algorithm has a worst-case running time in $2^{\bigO(\length{x})}$.
If, for some parameterization $\kappa$, we can give an algorithm of which the worst-case running time on any instance $x$ is in $2^{\bigO(\kappa(x))}\poly(\length{x})$ and, furthermore, we have that $\kappa(x) \ll \length{x}$ holds for at least \emph{some} arbitrarily large instances, then we can argue that $\kappa$ is a more accurate measure of the complexity of instances than is their size, since the running time of the second algorithm is exponential \emph{only} in the parameter value.
Note that this implies that interesting parameterizations cannot be monotonic functions.
More generally, for $X \subseteq \binary^\ast$ and a parameterization $\kappa$, a \emph{parameterized problem} $(X, \kappa)$ is said to be \emph{fixed-parameter tractable (fpt)} if, for some computable function $f$ and constant $c \geq 0$, there is an algorithm solving any instance $x$ of $X$ in time $f(\kappa(x))\length{x}^c$.\footnote{
  From here onward, we may write $k$ for $\kappa(x)$ when there is no risk of confusion.
  Also, $n$ stands for $\length{x}$ when specifying the complexity of an algorithm.
}
The essential feature of such running times is that the parameter value and instance size appear only in separate factors.

\paragraph*{Kernelization.}
An important notion in the study of fixed-parameter tractability is that of \emph{kernelization}.
Informally, a kernelization (or kernel) for a parameterized problem is a polynomial-time algorithm that, for any input instance, outputs an equivalent instance of which the size is upper-bounded by a function of the parameter.
This type of algorithm is usually presented as a formalization of preprocessing in the parameterized setting.
It reduces any instance with large size but small parameter value to an equivalent smaller instance, after which some other algorithm (possibly one with large complexity) is used to solve the reduced instance.
Another explanation, which fits well with the idea of studying the complexity of individual instances, is that a kernelization extracts the \emph{hard core} of an instance.

Of particular interest is the case where the upper bound on the size of the output instance of a kernelization is itself a polynomial function in the parameter.
Such \emph{polynomial kernelizations} are important because they offer a quick way to obtain efficient fpt-algorithms for a problem.
If $X$ is solvable in exponential-time, then the existence of a polynomial kernelization for $(X, \kappa)$ means that the problem can be solved in time $2^{\poly(k)}\poly(n)$, which roughly corresponds to what we might reasonably consider to be useful in practice.
Conversely, for many parameterized problems that can be solved by algorithms with such running times (for example, \pr{$k$-Vertex-Cover}), it is also possible to show the existence of polynomial kernelizations.
However, there are also exceptions, such as the \pr{$k$-Path} problem, where an algorithm with time complexity $2^{\bigO(k)}\poly(n)$, but no polynomial kernelization, is known.
It was a long-standing open question whether the existence of polynomial kernelizations is equivalent to having fpt-algorithms with a particular kind of running time.
Eventually, \citet{bodlaender2009problems} showed that for many fixed-parameter tractable problems (including \pr{$k$-Path}), the existence of polynomial kernels would imply the unlikely complexity-theoretic inclusion $\cl{NP} \subseteq \cladv{coNP}{poly}$.
This framework for proving conditional lower bounds against polynomial kernels was subsequently considerably extended and strengthened \citep{bodlaender2014kernelization,drucker2015new} (see also the survey of \citet{kratsch2014recent}).
In the same paper, \citeauthor{bodlaender2009problems} also unconditionally prove the existence of a parameterized problem that is solvable in time $\bigO(2^k n)$, but has no polynomial kernels, thus ruling out the possibility of an equivalence between polynomial kernels and fpt-algorithms with running times of the form $2^{\poly(k)}\poly(n)$.

\paragraph*{Generalized kernelization.}
A \emph{Turing kernelization} is an algorithm that can solve any instance of a parameterized problem in polynomial-time, provided it can query an oracle for the same problem  with instances of which the size is upper-bounded by a function of the parameter value of the input.
The idea here is that if we are willing to run an inefficient algorithm on an instance of size bounded in terms of the parameter alone (as was the case with \emph{regular} kernelizations), then we might as well run this algorithm on more than one such instance.
A regular kernelization can be regarded as a particular, restricted type of Turing kernelization that a) runs the polynomial kernelization algorithm on the input, b) queries the oracle for the resulting output instance, and c) outputs the oracle's answer.
As in the case of regular kernelizations, a \emph{polynomial Turing kernelization} is such that the bound on the size of the query instances is itself a polynomial function.

Polynomial Turing kernelizations are not as well-understood as regular kernels.
The methods for proving lower bounds against the size of regular kernels do not seem to apply to them.
Indeed, there are problems that most likely have no polynomial kernels, but which \emph{do} admit a polynomial Turing kernelization.
An example being \mbox{\pr{$k$-Leaf-Subtree}} (called \pr{Max-Leaf-Subtree} in \citep{cygan2016parameterized}).
Furthermore, there are only a few examples of non-trivial polynomial Turing kernelizations for problems that are not believed to admit polynomial regular kernelizations, such as restricted versions of \pr{$k$-Path} \citep{jansen2016turing,jansen2018turing} and of \pr{$k$-Independent Set} \citep{thomasse2017polynomial}.
Whether the general versions of these problems also have polynomial Turing kernels are major open questions in this field.

Compared to the regular kind, polynomial Turing kernelizations have a number of computational advantages, such as the ability to output the opposite of the oracle's answer to a query (non-monotonicity), the ability to make polynomially (in the size of the input) many queries, and the ability to adapt query instances based on answers to previous queries (adaptiveness).
Rather than focus on specific computational problems to determine the difference in strength between Turing and regular kernelizations, we instead look into the possibility of unconditionally separating the computational strengths of these two types of algorithms in general.
We investigate and answer a number of questions that, to our knowledge, were all open until now:
\begin{itemize}
\item
  Without relying on any complexity-theoretic assumptions, can we prove the existence of parameterized problems that admit polynomial Turing but not polynomial regular kernelizations?
  If so, which of the computational advantages of Turing kernelizations are sufficient for an unconditional separation?

  Note that for \mbox{\pr{$k$-Leaf-Subtree}}, only a larger number of queries is used, the known polynomial Turing kernel being both monotone and non-adaptive \citep[see][Section~9.4]{cygan2016parameterized}.
  On the other hand, the kernels in \citep{jansen2016turing} and \citep{thomasse2017polynomial} are adaptive.
\item
  Does every parameterized problem that is decidable in time $2^{\poly(k)}\poly(n)$, also admit a polynomial Turing kernelization?
\item
  To what extent can we relax the restrictions on regular kernelizations (viewed as Turing kernelizations), while still being able to apply known lower bound techniques?
  For example, can we rule out, for some natural problems, the existence of non-monotone kernels that make a few adaptive oracle queries?
\end{itemize}

\subsection{Overview of our results}
\begin{wrapfigure}[20]{r}{0pt}
  \begin{tikzpicture}
    \graph[grow up=1.363,nodes={anchor=center,align=center,inner xsep=1.4em},edges={double equal sign distance,-implies}]{
      ker/"polynomial kernels" --
      cTuring/"polynomial Turing kernels with\\a constant number of queries" --
      psize/"psize kernels" --
      tt/"polynomial\\truth-table kernels" --
      T/"polynomial\\Turing kernels" --
      super/"fixed-parameter tractable";
    };

    \coordinate (bottom left) at (current bounding box.south west);
    \path (current bounding box.north east) ++(0pt, -16pt) coordinate (top right);
    \pgfresetboundingbox;
    \useasboundingbox (bottom left) rectangle (top right);
  \end{tikzpicture}
  \caption{
    A hierarchy of polynomial kernels.
    Arrows signify a strict increase in computational power.
  }
  \label{fig:hierarchy}
\end{wrapfigure}
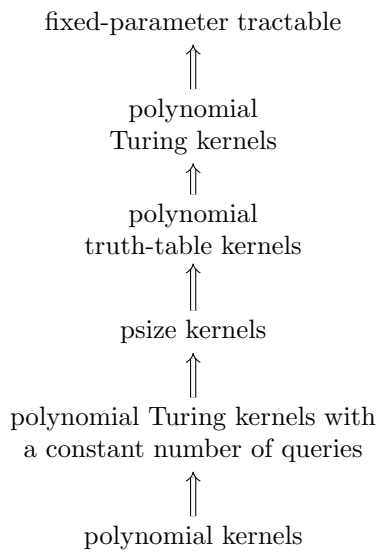

We show that each of the advantages of polynomial Turing kernelizations over polynomial regular kernelizations is, by itself, enough to unconditionally separate the two notions.
This produces a hierarchy of kernelizability within the class of problems that admit polynomial Turing kernelizations, Figure~\ref{fig:hierarchy}.
Specifically, we show that:
\begin{itemize}
\item there are problems that are not polynomially kernelizable, but do admit a polynomial Turing kernelization that makes a single oracle query (Theorem~\ref{thm:h1});
\item there are problems that admit non-adaptive polynomial Turing kernelizations (also known as polynomial \emph{truth-table} kernelizations), but cannot be solved by polynomial Turing kernelizations making a constant number of queries, even adaptively (Theorem~\ref{thm:htt} \& \ref{thm:hpsize});
\item there are problems that admit adaptive polynomial Turing kernelizations but not polynomial truth-table kernelizations (Theorem~\ref{thm:ht}).
\end{itemize}

Next, we show (Theorem~\ref{thm:htop}) that it is not enough for a problem to be decidable in time $2^{\poly(k)}\poly(n)$ in order for it to have a polynomial Turing kernelization.
In fact, the problem we construct can be solved in time $\bigO(2^k n)$.
Our theorem is stronger than a comparable result of \citeauthor{bodlaender2009problems}, who only exclude regular kernelizations.
We obtain a considerably simpler proof, harnessing the Time Hierarchy Theorem in favor of a direct diagonalization.

Finally, we ask how far up the hierarchy the known methods for proving lower bounds against polynomial kernelization can be applied.
The example of \mbox{\pr{$k$-Leaf-Subtree}} shows that they should already fail somewhere below polynomial truth-table kernelizations.
Indeed, we identify what we call \emph{psize kernelizations} as the apparently strongest type of polynomial Turing kernel that can be ruled out by current lower bound techniques (Section~\ref{sec:lowerbounds}).
A psize kernelization makes $\poly(k)$ non-adaptive oracle queries (of size $\poly(k)$), and then feeds the oracle's answers into a poly-sized circuit to compute its own final answer.
In terms of computational power, this type of kernelization stands between polynomial Turing kernelizations that make only a constant number of queries and polynomial truth-table kernelizations (Section~\ref{sec:separations}, Theorem~\ref{thm:htt} \& \ref{thm:hpsize}).

\subsection{Proof techniques}
The price we pay for being able to prove unconditional separations is that the problems we construct in the proofs are artificial rather than natural.
This is unavoidable, however, because computational problems that arise naturally will typically belong to classes that are hard to separate from $\cl{P}$ (such as $\cl{NP}$, $\cl{PH}$, $\cl{PP}$, etc.).
Thus, any claim that some parameterized version of a natural problem admits no polynomial kernelization, would currently have to rely on some complexity-theoretic assumptions.

In the construction of every problem witnessing a separation, diagonalization will be involved, in one way or another.
However, the application of diagonalization arguments in this context has some subtle issues.
An intuitive reason for this is the fact that it is very difficult to control the complexity of a problem that is constructed via an argument using diagonalization against polynomial-time machines.
Without additional complexity-theoretic assumptions, such problems can be forced to reside in powerful classes such as $\cl{EXP}$.
Positioning them in any interesting smaller classes is not straightforward.
By contrast, the difference between $\cl{P}$ and the class of problems that can be decided in polynomial-time with a very restricted form of access to an oracle, seems rather thin, and it is by no means clear whether a problem that is constructed via diagonalization can be placed between these two classes.
In Section~\ref{sec:separations} we discuss these issues, as well as how to overcome them, in detail.
Here, let us mention that the overall structure of our artificial problems resembles that of examples of natural problems which, subject to complexity-theoretic assumptions, admit polynomial Turing but not regular kernelizations.
Because of this, even the artificial examples we construct provide new insights into the power of Turing kernelization.

\section{Preliminaries}
We assume familiarity with standard notations and the basics of parameterized complexity theory, and refer the reader to \citep{flum2006parameterized} for the necessary background.
Here we review only the definitions of the notions most important for our work.

\begin{definition}
  A \emph{kernelization} (or \emph{kernel}) for a parameterized problem $(X, \kappa)$, where $X$ is a subset of $\binary^\ast$ and $\kappa$ is a parameterization, is a polynomial-time algorithm that, on a given input $x \in \binary^\ast$, outputs an instance $x' \in \binary^\ast$ such that $x \in X \Leftrightarrow x' \in X$ holds, and, for some fixed computable function $f$, we have $\length{x'} \leq f(\kappa(x))$.
  The function $f$ is referred to as the size of the kernel.
  The kernel is said to be \emph{polynomial} if $f$ is a polynomial.
\end{definition}

\begin{definition}
  A \emph{Turing kernelization} for a parameterized problem $(X, \kappa)$ is a polynomial-time algorithm that decides any instance $x$ of $X$ using oracle queries to $X$ of restricted size.
  For some fixed computable function $f$ that is independent of the input, the size of the queries must be upper bounded by $f(\kappa(x))$.
  A Turing kernelization is \emph{polynomial} if $f$ is a polynomial.

  A Turing kernelization is a \emph{truth-table kernelization} if, on every input, all of its oracle queries are independent of the oracle's answers.
  Thus, as an oracle machine, a truth-table kernelization is non-adaptive.
\end{definition}

A parameterized problem that exemplifies the relevance of our results is \mbox{\pr{$k$-Leaf-Subtree}}, where a graph $G$ and integer $k$ are given, and the question is whether $G$ has a subtree with at least $k$ leaves.
This problem admits a polynomial Turing kernelization but no polynomial regular kernelization, unless $\cl{NP} \subseteq \cladv{coNP}{poly}$.
See Section~9.4 of \citep{cygan2016parameterized} for a proof of the former, and Chapter~15 of the same reference for a proof of the latter fact.

\section{Separations}
\label{sec:separations}
To prove an unconditional separation between polynomial Turing kernelizability and polynomial regular kernelizability (or between two intermediate kinds of kernelizability), we construct a problem of which the instances can be solved in polynomial-time with oracle queries for small instances of the same problem.
We shall make sure that the instances cannot be solved in polynomial-time without such queries (remember, polynomial kernelizations are also poly-time decision procedures).
These requirements prevent us from constructing the classical part of our parameterized problem via simple diagonalization against polynomial-time machines.
The instances of the resulting language would not depend on each other in a way that would allow oracle queries to be useful, nor would all instances be solvable in time $p(n)$ for some \emph{fixed} polynomial $p$.
Solving an instance of such a language requires simulating Turing machines (\emph{TM}s) for a polynomial number of steps, but the degree of these polynomials increases with $n$.
Thus, a hypothetical polynomial Turing kernelization would neither be able to solve the instances of such a language directly within the allowed time, nor use its oracle access to speed up the computation.
An additional difficulty arises due to the bound on the size of the oracle queries (polynomial in $k$).
If the parameter value of an instance $x$ is too small relative to $\length{x}$, then the restricted oracle access of a  polynomial Turing kernelization may offer no computational advantage, since the instances for which the oracle can be queried will be small enough to be solved directly within the required time bound.

These issues can be overcome by designing a problem that shares what seems to be the essential feature of natural problems that, under complexity-theoretic assumptions, admit polynomial Turing but not polynomial (regular) kernelizations, such as the \mbox{\pr{$k$-Leaf-Subtree}} problem.
Recall that for this problem, a quadratic kernelization exists for the case when the input graph is connected, but that a polynomial kernelization for general graphs is unlikely to exist.
The known polynomial Turing kernelization for this problem works on general graphs by computing the kernel for each connected component of the input graph, and then querying the oracle for each of the $\bigO(n)$ resulting instances of size $\bigO(k^2)$ \citep[see][Section~9.4]{cygan2016parameterized}.
The crucial aspect here is that although the general problem may not admit polynomial kernelizations, it has a subproblem that does.
Furthermore, the polynomial Turing kernelization only queries instances of this subproblem.

The problems we construct will also have a polynomially kernelizable ``core,'' as well as a ``shell'' of instances that can be solved efficiently with small queries to the core.
Taking $V$ to be some decidable language, we can define
\begin{equation*}
  X(V) = \{\bits{0}x \suchthat x \in V\} \cup \left\{\bits{1}x \suchthat[\middle] \textrm{. . .}\right\},
\end{equation*}
where the ellipsis stands for a suitable condition that can be verified with small queries to $V$.
With the parameterization $\kappa$ such that $\kappa(\bits{0}x) = \length{x}$ and $\kappa(\bits{1}x) = \log\length{x}$ for all $x \in \binary^\ast$, the first set in the above disjoint union plays the role of the polynomially kernelizable core (it admits the trivial kernelization), while the second set plays the role of the shell.
The crucial observation now is that we can choose the condition that determines membership of an element of the form $\bits{1}x$ in $X(V)$ in such a way that a polynomial-time algorithm can decide the instance using small queries of the form $\bits{0}w$, \emph{regardless of the choice of $V$}.
Having thus secured the existence of a polynomial Turing kernelization (perhaps one that is further restricted), we are now free to construct $V$ via diagonalization against some weaker type of kernelization, so as to get the desired separation.

Using this approach, we prove that each of the computational advantages a polynomial Turing kernelization has over polynomial (regular) kernelizations, results in a strictly stronger type of kernelization, as shown in Figure~\ref{fig:hierarchy}.

\begin{theorem}
\label{thm:h1}
  There is a parameterized problem that has a polynomial Turing kernelization using only a single oracle query, but admits no polynomial kernelizations.
\end{theorem}
\begin{proof}
  Given any decidable set $V$, we can define
  \begin{equation*}
    X(V) = \{\bits{0}x \suchthat x \in V\} \cup \left\{\bits{1}x \suchthat[\middle] \log\length{x} \in \bbN\text{ and }\bits{0}^{\log\length{x}} \notin V\right\},
  \end{equation*}
  parameterized so that for all $x \in \binary^\ast$, $\kappa(\bits{0}x) = \length{x}$ and $\kappa(\bits{1}x) = \log\length{x}$.

  Clearly, the problem $(X(V), \kappa)$ has a polynomial Turing kernelization making a single query, regardless of the decidable set $V$.
For instances of the form $\bits{0}x$, the answer can be obtained by querying the oracle directly for the input, and if the input is $\bits{1}x$, one can query $\bits{0}^{\log\length{x} + 1}$ and output the opposite answer.

  We shall construct the set $V$ by diagonalization, ensuring that $X(V)$ does not admit a polynomial (regular) kernelization.
  Note that the kernelization procedures we diagonalize against can query $X(V)$, whereas we only decide the elements of $V$.
  Because every problem that admits a polynomial kernelization can also be decided by a polynomial-time TM that makes a single query of size $\poly(k)$ and then outputs the oracle's answer, we only need to diagonalize against this type of TM.
  As in a standard diagonalization argument, we run every such machine for an increasing number of steps, using as input the string $\bits{1}\bits{0}^{2^n}$ (the parameter value of which is $n$), where $n$ is chosen large enough for decisions made at previous stages to not interfere with the current simulation.
  Each machine is simulated until it runs out of time or makes an oracle query.
  Whenever the machine makes an oracle query different from $\bits{1}\bits{0}^{2^n}$, we answer it according to the current state of the set $V$.
  To complete the diagonalization, we either add $\bits{0}^n$ to $V$ or not, so as to ensure the machine's answer is incorrect.

  Note that for sufficiently large values of $n$, the string $\bits{1}\bits{0}^{2^n}$ cannot be queried, because $2^n$ outgrows any fixed polynomial in $n$ ($\in \poly(k)$).
  Additionally, a query to $\bits{0}\bits{0}^n$ is of no concern as the machine is incapable of negating the answer of the oracle.
\end{proof}

Next, we show that polynomial truth-table kernelizations, which can make $\poly(n)$ oracle queries of size $\poly(k)$ but cannot change their queries based on the oracle's previous answers, are more powerful than a restricted version of the same type of kernelization that makes at most $\poly(k)$ queries.
This restricted form of polynomial truth-table kernelization is of further interest because it can be ruled out by the current lower bounds techniques (see Section~\ref{sec:lowerbounds}).
We give the definition here.

\begin{definition}
\label{def:psize}
  A polynomial truth-table kernelization is a \emph{psize kernelization} if, on any input instance with parameter value $k$, it makes at most $\poly(k)$ oracle queries and its output can be expressed as the output of a $\poly(k)$-sized circuit that takes the answers of the oracle queries as input.
\end{definition}

The proof of the next theorem follows the same pattern as that of Theorem~\ref{thm:h1}, except that in the diagonalization part of the proof we now use the restriction on the number of queries the machines can make.
Recall that in Theorem~\ref{thm:h1} we made use of the machine's monotonicity, that is, the fact that its output must be equivalent to the outcome of its single oracle query.

\begin{theorem}
\label{thm:htt}
  There is a parameterized problem that has a polynomial truth-table kernelization but no psize kernelization.
\end{theorem}
\begin{proof}
  Given any decidable set $V$, we can define
  \begin{equation*}
    X(V) = \{\bits{0}v \suchthat v \in V\} \cup \left\{\bits{1}x \suchthat[\middle] \log\length{x} \in \bbN\text{ and }\binary^{\log\length{x}} \cap V \neq \emptyset\right\},
  \end{equation*}
  parameterized so that for all $x \in \binary^\ast$, $\kappa(\bits{0}x) = \length{x}$ and $\kappa(\bits{1}x) = \log\length{x}$.
  Clearly, $(X(V), \kappa)$ has a polynomial truth-table kernelization regardless of $V$: on input $\bits{0}x$ it queries the oracle for the input, and on input $\bits{1}x$, with $\log\length{x} \in \bbN$, it queries the oracle with each string $\bits{0}y$, for all $y \in \binary^{\log\length{x}}$, and accepts if one of the queries has a positive answer (otherwise it rejects).
  This procedure runs in polynomial time and makes at most $n$ oracle queries on any input of length $n+1$.

  We construct $V$ by diagonalizing against psize kernelization algorithms.
  To do this, we consider a computable list of TMs such that every machine appears infinitely often.
  At stage $i$ of the construction we choose $n$, a power of $2$, so that membership in $V$ has not been decided at a previous stage for any strings of length at least $\log n$.
  We then run the $i$-th machine on input $\bits{1}\bits{0}^n$ for $n^i$ steps.
  All new oracle queries are answered with `no', all other queries are answered so as to be consistent with previous answers.
  If the machine at stage $i$ terminates without having queried the oracle for all strings of the form $\bits{0}y$ with $y \in \binary^{\log n}$, we add an unqueried string of this length to $V$ if and only if the machine rejects.

  If $P$ is a psize kernelization, then the number of oracle queries it makes on an input $\bits{1}x$ is upper-bounded by $q(\log\length{x})$, for some fixed polynomial $q$.
  This is clearly $o(\length{x})$, so for some sufficiently large $i$ and $n$, $P$ will terminate without having queried all $n$ strings which can determine the correct answer.
  Thus, our diagonalization procedure will ensure that it terminates with the incorrect answer.
  On the other hand, the above-mentioned polynomial truth-table kernelization will always query all necessary strings in order to output the correct answer.
\end{proof}

Note that the conclusion of the above proof is actually that there exists a parameterized problem with a polynomial truth-table kernelization making $n-1$ oracle queries, that admits no polynomial (possibly adaptive!) Turing kernelization making fewer than $n-2$ queries on certain inputs of length $n$.
A psize kernel fits this condition, but is much more restricted (in particular, the number of allowed queries is polynomial in the parameter value).

Via a very similar proof, with a diagonalization argument relying on the number of oracle queries a machine can make, we can show that psize kernelizations are stronger than polynomial Turing kernelizations making any fixed finite number of queries, even adaptively.

\begin{theorem}
\label{thm:hpsize}
  There is a parameterized problem that has a psize kernelization but no polynomial Turing kernelization making only a constant number of (possibly adaptive) queries.
\end{theorem}

We can also show that adaptive queries provide a concrete computational advantage.
The proof of the separation between general polynomial Turing and truth-table kernelizations also follows the pattern of the previous three theorems, but with a more involved diagonalization argument, due to the need to distinguish between adaptive and non-adaptive oracle TMs.

\begin{theorem}
\label{thm:ht}
  There is a parameterized problem that has a polynomial Turing kernelization but no polynomial truth-table kernelization.
\end{theorem}

\begin{proof}
  For any decidable set $V$ we can define the function: $s^V: \binary^\ast \to \binary^\ast$ by
  \begin{equation*}
    s^V(q) = \begin{cases}
      \bits{0}q	& \text{if $q \notin V$}, \\
      \bits{1}q	& \text{if $q \in V$}.
    \end{cases}
  \end{equation*}
  Also for a decidable set $V$, we define the following parameterized problem:
  \begin{equation*}
    X(V) = \{\bits{0}x \suchthat x \in V\} \cup \left\{\bits{1}x \suchthat[\middle] \log\length{x} \in \bbN\text{ and } \underbrace{(s^V \circ s^V \circ \cdots \circ s^V)}_{(\log\length{x})^2\text{ times}}(\bits{0}^{\log{\length{x}}}) \in V\right\},
  \end{equation*}
  where the parameterization is defined so that for all $x \in \binary^\ast$, $\kappa(\bits{0}x) = \length{x}$ and $\kappa(\bits{1}x) = \log\length{x}$.
  The problem $X(V)$ has a polynomial Turing kernelization regardless of the set $V$: On inputs of the form $\bits{0}x$, the machine queries the oracle with its input (whose size is linear in the parameter value), and outputs the answer.
  On inputs of the form $\bits{1}x$ the machine makes the following $(\log\length{x})^2$ queries: $\bits{0}^{\log\length{x}+1}$, $\bits{0}b_1 \bits{0}^{\log\length{x}}$, $\bits{0}b_2 b_1 \bits{0}^{\log\length{x}},\ldots,\bits{0}b_{(\log\length{x})^2}\ldots b_1 \bits{0}^{\log\length{x}}$, where $b_i$ is the outcome of the $i$-th query, for each $i \le (\log\length{x})^2$.
  The output is the answer of the last oracle query.
  Since each of the queries in the second case is of size at most quadratic in $\kappa(1x) = \log\length{x}$, this procedure is a polynomial Turing kernelization.

  We now construct the set $V$ so that no polynomial truth-table kernelization can solve $X(V)$.
  Consider a variant of oracle TMs where the oracle can be queried for an arbitrary number of queries at once.
  Let $P_1, P_2, \ldots$ be a computable list of all such TMs in which each machine appears infinitely often.

  At each stage $i \in \bbN$, we set $n$ to be the smallest positive integer so that no oracle queries to $X(V)$ at any previous stage of the simulation depend on instances of $V$ of size at least $n$, and so that $n > i$ and $2^n > n^i$.
  At stage $i$ of the construction, we run $P_i$ on input $\bits{1}\bits{0}^{2^n}$ for $(2^n)^i$ steps (note that this is a polynomial of degree $i$ in $2^n+1$, the size of the input).
  In case the machine queries the oracle, let $S$ be the set of strings it queries.
  If $S$ includes strings of length at least $2^n$, we move on to the next stage.
  In particular, when no query of length $2^n + 1$ is made, $P_i$ is not making a query with prefix \bits{1} that is equivalent to the input.
  By the time bound, we have $\size{S} \leq 2^{ni}<2^{n^2}$, so there must be a string $y = b_{n^2}\ldots b_{2}b_{1}\bits{0}^{n}$, $b_j \in \binary$, such that  $\bits{0}y$ is not in $S$.
  The queries in $S$ are answered as follows: all queries also made at previous stages are answered so as to be consistent with previous answers; all queries of the form $\bits{0}b_j\ldots b_2b_1\bits{0}^n$, with $j \le n^2 - 1$, are answered with $b_{j+1}$; all other queries are answered with \bits{0} (`no').
  For all $j \le n^2 - 1$ such that $b_{j+1} = \bits{1}$, we place $b_j\ldots b_2b_1\bits{0}^n$ into $V$.
  After thus answering the queries in $S$, we resume the simulation of $P_i$ for the remainder of its allotted $2^{ni}$ steps and treat every subsequent invocation of the query instruction as a crash.
  Finally, we place $y$ into $V$ if and only if $P_i$ terminated within the time bound and rejected, making $\bits{1}\bits{0}^{2^n}$ a `yes'-instance if and only the $P_i$ rejects it.

  Assume now that there is a polynomial truth-table kernelization for $X(V)$.
  Such a procedure will eventually be targeted in the above construction.
  Indeed, a problem has a truth-table kernelization precisely when it is decided by a machine that runs in polynomial time and can make all its queries at once.
  Let $i$ be such that $P_i$ is a polynomial truth-table kernelization for $X(V)$, running in time $p(\length{x})$ on any input of the form $\bits{1}x$, and non-adaptively making oracle queries of size at most $q(\log\length{x})$, where $p$ and $q$ are fixed polynomials.
  As this machine occurs infinitely often in the list $P_1, P_2, \ldots$, we may assume that $i$ and its corresponding $n$ are large enough for $P_i$ to terminate on input $\bits{1}\bits{0}^{2^n}$, because we have $p(2^n+1) < 2^{ni}$.
  Moreover, we may assume that $i$ and $n$ are large enough for $q(n) < n^i < 2^n$ to hold.
  As $P_i$ will not be able to query all strings of the form $\bits{0}y\bits{0}^{n}$ with $\length{y} = n^2$, it will, by our construction of $V$, incorrectly decide some instance of $X(V)$.
\end{proof}

Finally, we show that decidability in time $2^{\poly(k)}\poly(n)$ does not guarantee the existence polynomial Turing kernelizations for the same problem.
This strengthens a theorem of \citet{bodlaender2009problems}, who construct a problem with the above complexity but rule out only polynomial regular kernelizations.

\begin{theorem}
\label{thm:htop}
  For every time-constructible function $g(k) \in 2^{o(k)}$, there is a problem that is solvable in time $\bigO(2^k n)$ but admits no Turing kernelization of size $g(k)$.
In particular, there is a problem that is solvable in time $\bigO(2^k n)$ but admits no polynomial Turing kernelization.
\end{theorem}
\begin{proof}
  Let $g(k)$ be a time-constructible function in $2^{o(k)}$.
  Without loss of generality, we may assume that $g(k)$ is also in $\Omega\left(2^{(\log k)^2}\right)$.
  Let $\kappa: \bbN \to \bbN$ be a time-constructible function such that we have $\kappa(n) \in \omega(\log n)\cap o(n)$ as well as $\kappa(g(k)) \in o(k)$ (for example, $\kappa(n) = \log n \log\left(\frac{g^{-1}(n)}{\log n}\right)$ is suitable).
  Let $t(n) = 2^{\kappa(n)}n$ and let $L$ be a language in $\mathbf{DTIME}(t(n))\setminus \mathbf{DTIME}(o(t(n)/\log(t(n)))$.
  Such a language exists by the Time Hierarchy Theorem.
  Assigning each instance $x$ of $L$ the parameter value $k = \kappa(\length{x})$, we find that $L$ can be solved in time $\bigO(2^k n)$.

  Furthermore, we have
  \begin{equation*}
    \frac{t(n)}{\log t(n)} = \frac{2^{\kappa(n)}n}{\kappa(n)+\log n} \in \Omega\left(2^{\kappa(n)}\right),
  \end{equation*}
  so we may conclude $2^{o(\kappa(n))}\subseteq o(t(n)/\log(t(n))$.

  Assume now that for some polynomial $p$, there exists a Turing kernelization for $L$ that runs in time $p(n)$ and queries the oracle with instances of size bounded by $g(k)$, where we set $k = \kappa(n)$.
  We show that such a Turing kernelization can be used to solve $L$ in time $o(t(n)/\log(t(n))$, contradicting the choice of the language.
  Our new algorithm will solve any instance $x$ with parameter value $k = \kappa(\length{x})$ by running the Turing kernelization on it, except that the instances for which the oracle is supposed to be queried are solved directly using the $\bigO(2^{\kappa(n)} n)$-time algorithm whose existence is guaranteed by the choice of $L$.
  The total running time of this new algorithm is then upper-bounded by:
  \begin{equation*}
    p(n) + p(n)2^{\kappa\left(g(k)\right)}g(k) = 2^{o(k)} = 2^{o(\kappa(n))},
  \end{equation*}
  which contradicts the lower bound on the deterministic time complexity of $L$.
\end{proof}

\section{Lower Bounds}
\label{sec:lowerbounds}
An immediate consequence of the separations arrived at in the previous section is that not all fixed-parameter tractable problems have polynomial kernelizations.
However, for any particular parameterized problem the (non-)existence of a polynomial kernelization may not be easy to establish.
The most fruitful program for deriving superpolynomial lower bounds on the size of regular kernelizations was started by \citet{bodlaender2009problems}.
While a straightforward application of their technique to Turing kernelizations is not possible, an extension to the psize level in our hierarchy, Figure~\ref{fig:hierarchy}, is feasible.

In order to keep our presentation focussed, we shall include only a limited exposition of the lower bound technique.
For a more complete overview, refer to \citep{downey2016fundamentals,kratsch2014recent}, or turn to \citep{bodlaender2014kernelization} for an in-depth treatment.
Central to the lower bounds engine are two similar looking classifications of instance aggregation.
The first of these does not involve a parameterization.
\begin{definition}
  A \emph{weak \algorithmicand-distillation} (\emph{weak \algorithmicor-distillation}) of a set $X$ into a set $Y$ is an algorithm that
  \begin{compactitem}
  \item receives as input a finite sequence of strings $x_1, x_2, \ldots, x_t$,
  \item uses time polynomial in $\sum_{i=1}^t \length{x_i}$,
  \item outputs a string $y$ such that
    \begin{compactitem}
    \item we have $y \in Y$ if and only if for all (any) $i$ we have $x_i \in X$,
    \item $\length{y}$ is bounded by a polynomial in $\max_{1 \le i \le t} \length{x_i}$.
    \end{compactitem}
  \end{compactitem}
\end{definition}

Note how the size of the output of a distillation is bounded by a polynomial in the \emph{maximum} size of its inputs and not by the sum of the input sizes.
Originally, distillations where considered where the target set $Y$ was equal to $X$, hence the \emph{weak} designator in this more general definition.
The parameterized counterpart to distillations is, as we shall soon see, more lenient than the non-parameterized one.
\begin{definition}
  An \emph{\algorithmicand-compositional} (\emph{\algorithmicor-compositional}) parameterized problem $(X, \kappa)$ is one for which there is an algorithm that
  \begin{compactitem}
  \item receives as input a finite sequence of strings $x_1, x_2, \ldots, x_t$  sharing a parameter value $k = \kappa(x_1) = \kappa(x_2) = \ldots = \kappa(x_t)$,
  \item uses time polynomial in $\sum_{i=1}^t \length{x_i}$,
  \item outputs a string $y$ such that
    \begin{compactitem}
    \item we have $y \in X$ if and only if for all (any) $i$ we have $x_i \in X$,
    \item $\kappa(y)$ is bounded by a polynomial in $k$.
    \end{compactitem}
  \end{compactitem}
\end{definition}

Here, a bound is placed on the \emph{parameter value} of the output of the algorithm, instead of on the \emph{length} of the output.
Additionally, this bound is a function of the \emph{unique} parameter value shared by all input strings.
Conceptually, a bound of this kind makes sense as parameter values serve as a proxy of the computational hardness of instances.
Thus, a parameterized problem is compositional, when instances can be combined efficiently, without an increase in computational hardness.

Generalizing the results of \citet{bodlaender2009problems,bodlaender2014kernelization}, we find that not just regular polynomial kernelizations, but also psize kernelizations tie the two ways of aggregating instances together.
For our proof to work, two aspects of the definition of psize kernelizations on page~\pageref{def:psize} that were not made explicit are crucial.
Firstly, because a psize kernelization is a \emph{polynomial} truth-table kernelization, the size of the queries can be bounded by a polynomial of the parameter value.
Secondly, it is important to note that the circuits involved must be uniformly computable from the input instances.

\begin{theorem}
  If $(X, \kappa)$ is an \algorithmicand-compositional (\algorithmicor-compositional) parameterized problem that has a psize kernelization, then $X$ has a weak \algorithmicand-distillation (weak \algorithmicor-distillation).
\end{theorem}

\begin{proof}
  Given a set $X$, consider the following set based on circuits and inputs derived from membership in $X$,
  \begin{equation*}
    C(X) = \{\langle\phi, (x_1, x_2, \ldots, x_t)\rangle \suchthat \text{$\phi$ is a circuit with $t$ inputs, accepting $(x_1 \in X, \ldots, x_t \in X)$}\}.
  \end{equation*}
  Note that a pairing of the specification of a circuit $\phi$ and $t$ strings $(x_1, x_2, \ldots, x_t)$ can be done so that $\length{\langle\phi, (x_1, x_2, \ldots, x_t)\rangle}$ is bounded by a polynomial in $\length{\phi} + \length{x_1} + \length{x_2} + \ldots + \length{x_t}$.

  We sketch the proceedings of a distillation that is given $x_1, x_2, \ldots, x_t$ as input.
  This procedure is adapted from \citep{bodlaender2009problems}.

  First, the inputs are grouped by their parameter value $k_i = \kappa(x_i)$ and the composition algorithm is applied to each group, obtaining $(y_1, k'_1), (y_2, k'_2), \ldots, (y_s, k'_s)$.
  Taking $k_\mathrm{max} = \max_{1 \le i \le t} k_i$, we have $s \le k_\mathrm{max}$ and, for some polynomial $p$, all $k'_i$ are bounded by $p(k_\mathrm{max})$.

  Next, the psize kernelization is applied to each $(y_i, k'_i)$, obtaining $s$ polynomial sized circuits and $s$ sequences of strings to query in order to get the inputs of the circuits.
  These circuits and strings can be amalgamated (dependent on the type of composition) into a single circuit $\phi$ and sequence of strings $(z_1, z_2, \ldots, z_r)$.

  We claim that the mapping of $(x_1, x_2, \ldots, x_t)$ to $\langle\phi, (z_1, z_2, \ldots, z_r)\rangle$ constitutes a weak distillation of $X$ into $C(X)$.
  Both $s$ and $k_\mathrm{max}$ are bounded by $\max_{1 \le i \le t} \length{x_i}$, since, for all $i$, we have $k_i \le \length{x_i}$.
  Therefore, the proposed weak distillation procedure produces an output of which the size is bounded by a polynomial in $\max_{1 \le i \le t} \length{x_i}$ and its running time is indeed polynomial in $\sum_{i=1}^t \length{x_i}$.
  Moreover, by definition of a psize kernelization the required preservation of membership is satisfied, hence the procedure is truly a weak distillation of $X$ into $C(X)$.
\end{proof}

Assuming we have $\cl{NP} \not\subseteq \cladv{coNP}{poly}$, it has been shown that \cl{NP}-hard problems admit neither weak \algorithmicor-distillations \citep{fortnow2011infeasibility}, nor weak \algorithmicand-distillations \citep{drucker2015new}.
Thus we can further our generalization of the results of \citet{bodlaender2014kernelization}.
\begin{corollary}
  If $(X, \kappa)$ is an \algorithmicand-compositional (\algorithmicor-compositional) parameterized problem and $X$ is \cl{NP}-hard, then $(X, \kappa)$ does not have a psize kernelization unless $\cl{NP} \subseteq \cladv{coNP}{poly}$.
\end{corollary}

Accordingly, our hierarchy of polynomial kernels is not merely synthetic and the place of many natural problems in the hierarchy is lower bounded.
In light of the more general setting of \citet{bodlaender2014kernelization}, we remark that a generalization of our results to cross-composition (generalizing compositionality) and psize compression (generalizing psize kernelization) is immediate.

\section{Classical Connections}
Algorithms for fixed-parameter tractable problems are not easily diagonalized against.
Such algorithms have a running time of the form $f(\kappa(x))\length{x}^c$, where $f$ is a computable function and $c$ a constant.
The challenge in diagonalizing is caused by the absence of a computable sequence of computable functions such that every computable function is outgrown by a member of the sequence.
However, as witnessed by this document, diagonalization can be used to uncover structure \emph{inside} \cl{FPT}.
Key to this possibility is that a problem is fixed-parameter tractable precisely when it is kernelizable, and the running time bound for kernelizations does not include arbitrary computable functions.

While, to our knowledge, not done before in a parameterized context, separating many--one, truth-table, and Turing reductions is an old endeavour, dating back to \citet{ladner1975comparison}.
Indeed, kernelizations are in essence reductions, more specifically, they are \emph{autoreductions} in the spirit of \citet{trakhtenbrot1970autoreducibility}.
Since kernelizations come with a time bound, a Turing kernelization could more accurately be described as a \emph{bounded Turing} kernelization, or \emph{weak truth-table} kernelization \citep[see][Section~3.8]{soare2016turing}.
However, the adaptiveness of a Turing kernelization entails that the number of different queries it \emph{could} make (unaware of the answers of the Oracle) is much higher than that of a truth-table kernelization, given the same time bound.
In that sense, our separation based on adaptiveness, Theorem~\ref{thm:ht}, is also a separation based on the number of queries made.

An important feature of kernelizations is not covered by an interpretation of kernelizations as autoreductions.
Where the definition of an autoreduction excludes querying the input string, the definition of a kernelization imposes a stronger condition on the queries, namely a size bound as a function of the parameter value.
In this light, it may be worthwhile comparing kernelizations to a more restrictive type of autoreduction, the \emph{self-reduction} \citep[see][Section~4.5]{balcazar1995structural}.
Self-reducibility is defined in \citep{balcazar1995structural} as autoreducibility where all queries are shorter than the input.
However, many of the results around self-reducibility extend to more general orders than the ``shorter than''-order and the definition can be generalized \citep{ko1983self}.
While the size bound on the queries that is required of kernelizations does not fit the self-reducibility scheme perfectly, the similarities in the definitions urge the consideration of other forms of self-reducibility in a parameterized context.
In particular, reducibility with a decreasing parameter value may be of interest.

\bibliographystyle{plainnat}

\end{document}